\documentclass{article}

\usepackage{footmisc}
\usepackage{amsmath,amssymb,amsfonts}
\usepackage{amssymb,amsfonts}
\usepackage[latin1]{inputenc}
\usepackage[english]{babel}
\usepackage{psfrag}
\usepackage{times}
\usepackage{graphicx,epsfig,multicol}
\usepackage{epsfig}
\usepackage[latin1]{inputenc}
\usepackage{bm}
\usepackage{tikz}
\usepackage{amsthm,amsfonts,amsmath,amssymb}

\newtheorem{remark}{\bf Remark}
\newtheorem{theorem}{\bf Theorem}
\newtheorem{lemma}{\bf Lemma}
\newtheorem{property}{\bf Property}
\newtheorem{definition}{\bf Definition}
\newtheorem{corollary}{\bf Corollary}
\newtheorem{proposition}{\bf Proposition}

\def\alphadot{{\dot \alpha}}
\def\deltadot{{\dot \delta}}

\def\R{\mathbb{R}}
\def\N{\mathbb{N}}

\newcommand{\rr}{\mathbf{r}}

\newcommand{\x}{\mathbf{x}}

\newcommand{\bxi}{\bm{\xi}}
\newcommand{\bzero}{\underline 0}

\newcommand{\C}{\mathbb{C}}

\newcommand{\angmom}{{\bf c}}
\newcommand{\energy}{{\cal E}}
\newcommand{\lenz}{{\bf L}}

\newcommand{\erre}{{\bf r}}
\newcommand{\erredot}{\dot{\bf r}}

\newcommand{\erho}{\hat{\bf e}^\rho}
\newcommand{\ealpha}{\hat{\bf e}^\alpha}
\newcommand{\edelta}{\hat{\bf e}^\delta}
\def\eort{{\bf e}^\perp}

\def\DD{{\bf D}}
\def\EE{{\bf E}}
\def\FF{{\bf F}}
\def\GG{{\bf G}}
\def\JJ{{\bf J}}
\def\cc{{\bf c}}
\def\rhodot{\dot{\rho}}
\def\rr{{\bf r}}
\def\qu{{\bf q}}  
\def\qudot{\dot{\bf q}} 

\def\Att{{\cal A}}

\def\Avec{{\bf A}}

\def\Rvec{{\bf R}}
\def\Rcal{\mathcal{R}}

\def\bzero{{\bf 0}}   
\def\bDelta{\bm{\Delta}}

\def\alphadot{\dot{\alpha}}
\def\deltadot{\dot{\delta}}

\def\erho{{\bf e}^\rho}
\def\ealpha{{\bf e}^\alpha}
\def\edelta{{\bf e}^\delta}

\def\bq{{\bf q}}
\def\br{{\bf r}}

\def\rhodot{\dot{\rho}}

\def\bv{{\bf v}}
\def\bzero{{\bf 0}}

\def\bqdot{\dot{\bf q}}
\def\newu{\mathfrak{v}}
\def\oldu{\mathfrak{u}}

\def\bPhi{\bm{\Phi}}
\def\bPsi{\bm{\Psi}}
\def\etabf{\bm{\eta}}
\def\bu{{\bf u}}

\begin{document}
%
\title{Orbit determination with the Keplerian integrals}
%
%
\author{{\bf Giovanni Federico Gronchi}\\
  \\
%
%
%
%
  Dipartimento di Matematica, Universit\`a di Pisa\\
  Largo B. Pontecorvo 5, 56127, Pisa, Italy
}

\maketitle              

\begin{abstract}
  We review two initial orbit determination methods for too short
  arcs (TSAs) of optical observations of a solar system body. These
  methods employ the conservation laws of Kepler's problem, and allow
  to attempt the linkage of TSAs referring to quite far epochs,
  differing by even more than one orbital period of the observed object. The
  first method ({\tt Link2}) concerns the linkage of 2 TSAs, and leads
  to a univariate polynomial equation of degree 9.  An optimal
  property of this polynomial is proved using Gr\"obner bases theory.
  The second method ({\tt Link3}) is thought for the linkage of 3
  TSAs, and leads to a univariate polynomial equation of degree 8. A
  numerical test is shown for both algorithms.
  
\end{abstract}

\rightline{\em dedicated to Prof. Andrea Milani}

\section{Introduction}

Modern telescopes collect a very large number of optical observations
of solar system bodies, that can be usually grouped in {\bf very short
  arcs} (VSAs), see \cite{mg10}.  A VSA is a set
\[
  \{(\alpha_i,\delta_i),\quad  i=1\ldots m\}, \qquad m\geq 2
\]
of pairs of values of {\em right ascension} and {\em declination} of
the same celestial body, referring to epochs $t_i$, and covering a
very short path in the sky.  Usually the data contained in a VSA do
not allow to compute a least squares orbit: in this case we speak of a
too short arc (TSA).  Given a TSA, we can compute an {\bf
  attributable}
\[
\Att = (\alpha,\delta,\alphadot,\deltadot)
\]
at the mean epoch $\bar t = \frac{1}{m}\sum_i t_i$ of the observations
by a linear or quadratic fit, see \cite{mg10}.
Given an
attributable, the radial distance $\rho$ and the radial velocity
$\dot{\rho}$ of the observed body remain completely unknown.  However,
given two attributables referring to the same celestial body, we can
try to put them together with the aim of computing an orbit that fits
all the data. This operation is called {\bf linkage} in the
orbit determination literature, and it is often challenging:
an orbit produced by
linking together two TSAs usually needs a confirmation with additional
data to be considered reliable.  Moreover, we cannot know {\em a
  priori} that two TSAs refer to the same observed body, and to
perform an efficient selection of pairs of TSAs to be passed to a
linkage algorithm is a critical issue. 

In this paper we review two recent initial orbit determination methods,
introduced in \cite{gbm15} and \cite{gbm17}, for the linkage of two or
three TSAs. These are called {\tt Link2} and {\tt Link3}, respectively. Some interesting algebraic aspects of these algorithms are
also discussed, and a numerical test is shown for both.
  
\section{Linkage with the Keplerian integrals}
\label{s:kepintlink}

The first integrals of Kepler's motion can be used to write polynomial
equations for the linkage of 2 TSAs.  The conservation laws of angular
momentum and energy were proposed for the linkage problem already in
 \cite{th77},
\cite{taff84}, \cite{trs84}: here the authors observed that the equations could be put
in polynomial form but did not use this form, see.
A polynomial formulation of the linkage problem was considered later
in a series of papers \cite{gdm10}, \cite{gfd11}, \cite{gbm15}.  In
\cite{gdm10} the angular momentum and energy conservation laws are
used, as in \cite{th77}: a polynomial is obtained by squaring twice
the equation of the energy conservation. After elimination of
variables we get a univariate equation of degree 48 in the radial
distance $\rho_2$.  In \cite{gfd11} the degree is reduced to 20 by
using the Laplace-Lenz vector projected along a suitable direction in
place of the energy. In \cite{gbm15} all the algebraic conservation
laws are combined so that the degree is reduced to 9: this is the
algorithm that we recall here.


\begin{remark}
Classical preliminary orbit determination methods, e.g. the ones by
Gauss, Laplace, Mossotti \cite{gauss1809}, \cite{laplace},
\cite{cellpinz}, \cite{gbrjm21} use the equations of motion, and
Taylor series expansions around a central time of the observational arc, thus
the observations must necessarily be close enough in time. We observe
that using conservation laws this constraint on the time is not
required.
\end{remark}
  
\subsection{Kepler's problem and its first integrals}

The equation of motion of Kepler's problem is
\begin{equation}
\ddot{\rr} = -\mu\displaystyle\frac{\rr}{|\rr|^3},
\label{kepeq}
\end{equation}
where $\rr\in\R^3$ is the unknown position vector and $\mu$ is a
positive constant.
The dynamics defined by \eqref{kepeq} has the following conserved quantities:
\begin{eqnarray*}
  &&{ \cc = \rr \times \dot{\rr}},\hskip 1.3cm
  \mbox{\it angular momentum}\\
  &&{ {\cal E} = \frac{1}{2}|\erredot|^2 -
      \frac{\mu}{|\erre|}},\hskip 0.6cm \mbox{\it energy}\\
  &&{ \lenz = \frac{1}{\mu}\erredot\times\angmom -
              \frac{\erre}{|\erre|}}, \hskip
  0.3cm \mbox{\it Laplace-Lenz vector}.
\end{eqnarray*}
%
%
We call these quantities the {\bf Keplerian integrals}.
Since $\angmom$ and $\lenz$ have 3 components we get 7 scalar
conserved quantities: among them only 5 are independent, in fact 
\[
    \angmom\cdot\lenz = 0,\qquad
     2|\angmom|^2\energy + {\mu}^2(1-|\lenz|^2) = 0.
\]

Given an attributable $\Att$ at the epoch $\bar t$, we write below the
Keplerian integrals as functions of the unknown radial distance and
velocity $\rho,\rhodot$.
We start by writing
\[
\begin{split}
  &\erre = \bq + \rho\,\erho,\cr
  &\erredot = \bqdot + \rhodot\,\erho + \rho(\alphadot\cos\delta\ealpha + \deltadot\edelta),
\end{split}
\]
where $\bq$, $\bqdot$ are the position and velocity of the observer at
time $\bar t$,
\[
\erho = (\cos\delta\cos\alpha, \cos\delta\sin\alpha, \sin\delta)
\]
gives the line of sight, and
\[
\ealpha =
(\cos\delta)^{-1}\frac{\partial\erho}{\partial\alpha}, \qquad\edelta =
\frac{\partial\erho}{\partial\delta}. 
\]
  
The angular momentum vector can be expressed as
\[
\cc(\rho,\rhodot) = \rr \times \dot{\rr} = \DD \rhodot + \EE
\rho^2 + \FF \rho + \GG.
\]
The vectors $\DD,\EE,\FF,\GG$ depend only on the attributable $\Att$
and on $\bq, \bqdot$:
\begin{equation}
\begin{array}{l}
\DD = \bq\times\erho,\cr
\EE = \alphadot\cos\delta\erho\times\ealpha +
\deltadot\erho\times\edelta
= \alphadot\cos\delta\edelta - \deltadot\ealpha ,\cr
\FF = \alphadot\cos\delta\bq\times\ealpha + \deltadot\bq\times\edelta + \erho
\times\bqdot,\cr
\GG = \bq\times\bqdot. \cr
\end{array}
\label{DEFG}
\end{equation}

The { energy} can be written as
\[
{\cal E} = \frac{1}{2}|\erredot|^2 -  \frac{\mu}{|\erre|},
\]
where
\[
|\erre| = \sqrt{\rho^2 + 2(\bq\cdot\erho)\rho + |\bq|^2 },
\]
and
\[
\vert\erredot\vert^2 = \rhodot^2 
+ (\alphadot^2\cos^2\delta + \deltadot^2)\rho^2 + 2\bqdot\cdot\erho\rhodot
+ 2\bqdot\cdot(\alphadot\cos\delta\ealpha + 
\deltadot\edelta)\rho + \vert\bqdot\vert^2.
\]
Finally, the Laplace-Lenz vector $\lenz$ is given by
\[
\mu\lenz(\rho,\rhodot) = 
\Bigl(\vert\erredot\vert^2 - \frac{\mu}{|\erre|} \Bigr)\erre -
(\erredot\cdot\erre)\erredot,
\]
where
\[
\erredot\cdot\erre = \rho\rhodot + \bq\cdot\erho\rhodot + 
(\bqdot\cdot\erho + \bq\cdot\ealpha\alphadot\cos\delta +
\bq\cdot\edelta\deltadot)\rho +  \bqdot\cdot\bq .
\]

\begin{remark}
  The expressions of $\energy$ and $\lenz$ are algebraic but not
  polynomial, due to the presence of the term $\mu/|\erre|$. If we
  consider the auxiliary variable $z$ defined by relation
  \begin{equation}
    { z|\erre| = \mu,}
    \label{urhodep}
  \end{equation}
  then the Keplerian integrals can be viewed as polynomials in the
  variables $\rho,\rhodot,z$ by writing $z$ in place of $\mu/|\erre|$.
  In this way, we obtain
  \[
  \tilde\energy = \frac{1}{2}|\erredot|^2 - z,
  \hskip 1cm
  \mu\tilde\lenz = (|\erredot|^2 - z)\erre - (\erredot\cdot\erre)\erredot.
  \]
  The relation between $\rho$ and $z$ can be taken into account through
  the polynomial equation
  \begin{equation}
    |\erre|^2z^2 = \mu^2.
    \label{zintro}
  \end{equation}
  \label{rem:z}
  Moreover, the following relations hold:
  \[
  \angmom\cdot\tilde{\lenz} = 0,
  \hskip 1cm
  2|\angmom|^2\tilde{\energy} + \mu^2(1-|\tilde{\lenz}|^2) = 0.
  \]
\end{remark}

\subsection{Polynomial equations for the linkage}
  
Given two attributables $\Att_1, \Att_2$ at the epochs $\bar{t}_1, \bar{t}_2$,
referring to the same solar system body, we consider the system
\begin{equation}
{ \angmom_1 =\angmom_2, \quad
\lenz_1 = \lenz_2, \quad
\energy_1 = \energy_2,
}
\label{fulls}
\end{equation}
of 7 algebraic (but not polynomial) equations in the 4 unknowns
$\rho_1$, $\rho_2$, $\rhodot_1$, $\rhodot_2$.
%
%
System \eqref{fulls} depends on the vector of known parameters
\[
(\Att_1, \Att_2, \bq_1, \bq_2, \bqdot_1, \bqdot_2),
\]
and is overdetermined.

If we assume that the two-body dynamics is perfectly respected, and no
error occurs in the coefficients, then the set of solutions of
\eqref{fulls} in the complex field $\C$ (but also in $\R$) is not
empty.  More realistically, since these assumptions cannot hold
exactly, system \eqref{fulls} is generically\footnote{i.e. such
  property can not be violated in a non-empty open subset of the data
  set $\bq_j,\bqdot_j,\Att_j$, $j=1,2$} inconsistent.

Combining the equations in \eqref{fulls} we can obtain an
overdetermined polynomial system which is consistent and can be reduced by
elimination to a univariate polynomial $\mathfrak{u}$ of degree 9 in
one of the radial distance, e.g. $\rho_2$, as will be shown below.




The conservation of the angular momentum $\angmom_1=\angmom_2$
can be written as
\begin{equation}
\DD_1\rhodot_1-\DD_2\rhodot_2 = \JJ(\rho_1,\rho_2),
\label{ameq}
\end{equation}
where
\begin{equation}
\JJ(\rho_1,\rho_2) = \EE_2\rho_2^2 - \EE_1\rho_1^2 + \FF_2\rho_2 -
\FF_1\rho_1 + \GG_2 - \GG_1,
\label{JJ}
\end{equation}
and $\DD_j,\EE_j,\FF_j,\GG_j$ are given by relations \eqref{DEFG} at
times $\bar{t}_j$.
Projecting equations \eqref{ameq} onto the vectors
\[
\DD_1\times\DD_2,\quad 
\DD_1\times(\DD_1\times\DD_2),\quad 
\DD_2\times(\DD_1\times\DD_2),
\]
where
\[
\DD_j = \bq_j\times\erho_j,
\]
we get
\begin{eqnarray}
&&\JJ(\rho_1,\rho_2)\cdot(\DD_1\times\DD_2) = 0,\nonumber\\%
&&|\DD_1\times\DD_2|^2\rhodot_1 - \JJ(\rho_1,\rho_2)\cdot
  \DD_1\times(\DD_1\times\DD_2) = 0,\label{rho1dot}\\
&&|\DD_1\times\DD_2|^2\rhodot_2 - 
\JJ(\rho_1,\rho_2)\cdot \DD_2\times(\DD_1\times\DD_2) = 0.\label{rho2dot}
\end{eqnarray}
We set
\[
q(\rho_1,\rho_2) = \JJ(\rho_1,\rho_2)\cdot(\DD_1\times\DD_2).
\]
This is a quadratic polynomial, that can be written as
\begin{equation}
q(\rho_1,\rho_2) = q_{2,0}\rho_1^2 + q_{1,0}\rho_1 +
q_{0,2}\rho_2^2 + q_{0,1}\rho_2 + q_{0,0},
\label{qpoly}
\end{equation}
with
\[
\begin{array}{l}
q_{2,0} = -\EE_1\cdot\DD_1\times\DD_2,\cr
q_{1,0} = -\FF_1\cdot\DD_1\times\DD_2,\cr
\end{array}
\hskip 1cm
\begin{array}{l}
q_{0,2} = \EE_2\cdot\DD_1\times\DD_2,\cr
q_{0,1} = \FF_2\cdot\DD_1\times\DD_2,\cr
\end{array}
\]
\[
q_{0,0} = (\GG_2-\GG_1)\cdot\DD_1\times\DD_2.
\]
\begin{remark}
Using equations (\ref{rho1dot}), (\ref{rho2dot}) we can write
$\rhodot_1, \rhodot_2$ as quadratic polynomials in the variables
$\rho_1,\rho_2$.  This corresponds to using conservation of
angular momentum in the plane orthogonal to $\DD_1\times\DD_2$.
\end{remark}

The equations
\begin{equation}
  \lenz_1=\lenz_2, \qquad \energy_1 = \energy_2
  \label{energylenz}
\end{equation}
are algebraic but not polynomial, due to the terms ${ \mu/|\erre_j|}$.
We consider the equation
\begin{equation}
\bxi = \bzero,
\label{xi=0}
\end{equation}
with
\begin{eqnarray}
  \bxi
&=& [\mu(\lenz_1-\lenz_2) -
  (\energy_1\erre_1-\energy_2\erre_2)]\times(\erre_1-\erre_2)\label{xi}
\label{xidef}\\
&=& \frac{1}{2}(|\erredot_2|^2- |\erredot_1|^2)\erre_1\times\erre_2
- (\erredot_1\cdot\erre_1)\erredot_1\times(\erre_1-\erre_2) +
(\erredot_2\cdot\erre_2)\erredot_2\times(\erre_1-\erre_2) .\nonumber
\end{eqnarray}
Note that in equation \eqref{xi=0}, which is a consequence of
\eqref{energylenz}, the dependence on $\mu/|\erre_j|$ has been
canceled.

After eliminating $\rhodot_1$, $\rhodot_2$ by \eqref{rho1dot},
\eqref{rho2dot}, $\bxi$ becomes a bivariate vector polynomial with total degree 6,
that we still denote by $\bxi$. In the following, we consider the bivariate polynomial system
\begin{equation}
q = 0, \qquad \bxi = \bzero,
\label{qxi}
\end{equation}
which is a consequence of \eqref{fulls}.

\begin{remark}
The monomials of $\bxi$ with the highest
degree are all multiplied by $\erho_1\times\erho_2$. Therefore, the
two projections
\begin{equation}
p_1 = \bxi\cdot\erho_1,
\hskip 0.7cm
p_2 = \bxi\cdot\erho_2
\label{p1p2}
\end{equation}
lower the degree, and give two polynomials with total degree 5.
\end{remark}


\subsection{Consistency of equations \eqref{qxi}} 
  
We sketch the proof of the following result: full details are given in
\cite{gbm15}.
\begin{theorem}
For generic values of the data, the bivariate and overdetermined polynomial system
\[
q = 0, \qquad \bxi = \bzero
\]
is consistent.  Moreover, it can be reduced by elimination to a system of two
univariate polynomials of degree 10,
\begin{equation}
\mathfrak{u}_1 = \mathfrak{u}_2 = 0,
\label{u1u2}
\end{equation}
such that
\[
\mathfrak{u}=\mathrm{gcd}(\mathfrak{u}_1,\mathfrak{u}_2)
\] 
has degree 9.
\label{udeg9}
\end{theorem}

\noindent {\em Sketch of the proof.}
Generically, relation
\[
(\bq_1-\bq_2)\cdot\erho_1\times\erho_2\neq 0
\]
holds, so that, from
\[
\bxi\cdot(\br_1-\br_2)=0\qquad \mbox{and}\qquad \br_j=\bq_j + \rho_j\erho_j\quad (j=1,2),
\]
we obtain
\[
\bxi = \bzero \qquad \mbox{if and only if} \qquad \bxi\cdot\erho_1 =
\bxi\cdot\erho_2 = 0.
\]
Thus, in place of \eqref{qxi} we can consider the bivariate and still
overdetermined system
\begin{equation}
q = p_1 = p_2 = 0.
\label{qp1p2}
\end{equation}
We note that,
if $(\rho_1, \rho_2)$ fulfills $q=0$, the vectors $\erre_1$,
$\erre_2$, $\erredot_1$, $\erredot_2$ all lie in the same plane.
This remark leads to the following geometrical fact:
\begin{property} For $(\rho_1, \rho_2)$ fulfilling $q=0$ the vector $\bxi$ is
parallel to the common value $\angmom=\angmom_1=\angmom_2$ of the
angular momentum.
\label{geomrem}
\end{property}
Each projection
\[
  p_j=\bxi\cdot\erho_j, \qquad j=1,2
\]
vanishes either if $\bxi=\bzero$, or if $\bxi$ is orthogonal to
$\erho_j$. By Property~\ref{geomrem}, when $q=0$ relation $\bxi\cdot\erho_j=0$ can
be checked using the angular momentum in place of $\bxi$.  For this purpose we introduce the
projections
\[
c_{ij} = \angmom_i\cdot\erho_j, \qquad i,j=1,2.
\]
The equations
\[
c_{11}(\rho_1,\rho_2) = 0, \qquad c_{22}(\rho_1,\rho_2) = 0
\]
define straight lines, while
\[
c_{12}(\rho_1,\rho_2) = 0, \qquad c_{21}(\rho_1,\rho_2) = 0
\]
define conic sections, see Figure~\ref{fig:cij}.

\begin{figure}[h!]
  \begin{center}
\begin{tikzpicture}
  \coordinate (O) at (0,0);
  \coordinate (xl) at (-1,0);
  \coordinate (xr) at (4.5,0);
  \coordinate (yd) at (0,-1.5);
  \coordinate (yu) at (0,3);
  
  \draw[-latex] (xl)--(xr) node [below] {$\rho_1$};
  \draw[-latex] (yd)--(yu) node [left] {$\rho_2$};

  \coordinate (c11d) at (3.07,-1.5);
  \coordinate (c11u) at (3.07,3);
  \coordinate (c22l) at (-1,1.57);
  \coordinate (c22r) at (4.5,1.57);
  \draw[dashed] (c11d)--(c11u);
  \draw[dashed] (c22l)--(c22r);
  \draw (c11u) node [right] {$c_{11}=0$};
  \draw (c22r) node [above] {$c_{22}=0$};

  \draw (1.8,2.9) node [above,color=blue] {$c_{12}=0$};
  \draw (4.,1) node [below,color=red] {$q=0$};
  \draw (1.5,-1.1) node [below,color=green] {$c_{21}=0$};
  
  \draw[scale=2.,color=red] (0.9,0.3) ellipse (25pt and 20pt); 
  \draw[scale=2.,color=blue] (0.9,1.02) ellipse (22pt and 12pt); 
  \draw[scale=2.,color=green] (0.715,0.3) ellipse (28pt and 25pt); 

  \coordinate (C) at (3.07,1.57);
  \draw (C) node [right] {$C$};
  \fill (C) circle (0.3mm);
  
  \coordinate (P1) at (3.07,-0.37);
  \draw (P1) node [right] {$P_1$};
  \fill (P1) circle (0.3mm);
  \draw[dotted] (-1,-0.37)--(4,-0.37);

  \draw (0.1,-0.3) node[left] {$\rho_2'$};
  \draw (0.1,1.75) node[left] {$\rho_2''$};
  
  \coordinate (P2) at (0.53,1.57);
  \draw (P2) node [above] {$P_2$};
  \fill (P2) circle (0.3mm);
  \draw[dotted] (0.53,-1.5)--(0.53,3);
  
  \draw (0.53,0.1) node[below] {$\rho_1'$};
  \draw (2.9,0.1) node[below] {$\rho_1''$};

\end{tikzpicture}
\end{center}
  \caption{Curves given by $q=0$, $c_{ij}=0$.}
  \label{fig:cij}
\end{figure}
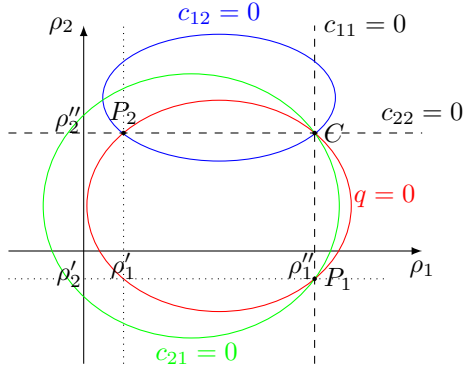


Set
\[
P_1 = (\rho_1'',\rho_2'), \qquad P_2 = (\rho_1',\rho_2''),\qquad C =
(\rho_1'',\rho_2''),
\]
where
\begin{equation}
\rho_1' = \frac{\qu_1\times\qu_2\cdot\erho_2}{\erho_1\times\erho_2\cdot\qu_2},
\qquad
\rho_2' = \frac{\qu_1\times\qu_2\cdot\erho_1}{\erho_1\times\erho_2\cdot\qu_1},
\label{rhojprimo}
\end{equation}
\begin{equation}
\rho_1'' =
\frac{\qu_1\times\qudot_1\cdot\erho_1}{\erho_1\times\eort_1\cdot\qu_1},
\qquad \rho_2'' =
\frac{\qu_2\times\qudot_2\cdot\erho_2}{\erho_2\times\eort_2\cdot\qu_2},
\label{rhojsecondo}
\end{equation}
with
\[
\eort_j = \alphadot_j\cos\delta_j\ealpha_j + \deltadot_j\edelta_j,
\qquad j=1,2.
\]
These points fulfill the relations
\[ 
\begin{split}
  &c_{11}(P_1)= q(P_1) = 0,\\
  &c_{22}(P_2) = q(P_2) = 0,\\
  &c_{11}(C) = c_{22}(C) = q(C) = 0.
\end{split}
\]

We use the following results, that hold generically, see \cite{gbm15}.

\begin{lemma}
  The point $C=(\rho_1'',\rho_2'')$
  satisfies
  \[{ \angmom_1(C) = \angmom_2(C)=\bzero}
  \]
  and $C$ is the unique point in the plane $\rho_1\rho_2$ where both
  angular momenta vanish.
\label{lemma1}
\end{lemma}

\begin{lemma}
  In $C$ we have $\bxi\cdot\erho_1\neq 0$
  and $\bxi\cdot\erho_2\neq 0$.
  \label{lemma2}
\end{lemma}

\begin{lemma}
  Assume $q=0$. Then $\bxi=\bzero$ is equivalent to
  \begin{equation}
    {\left\{
      \begin{array}{l}
        \bxi\cdot\erho_1 =0\cr
        \angmom\cdot\erho_1 \neq 0\cr
      \end{array}
      \right.
      \quad \mbox{ or }\quad
      \left\{
      \begin{array}{l}
        \bxi\cdot\erho_2 = 0\cr
        \angmom\cdot\erho_2 \neq 0\cr
      \end{array}
      \right.
      .}
    \label{alternative}
  \end{equation}
  \label{lemma3}
\end{lemma}

\noindent Using these lemmas, first we show that system \eqref{qp1p2}
has {\em at least} 9 solutions (in the complex field $\C$).
By Lemma~\ref{lemma3}, system \eqref{qp1p2} is generically equivalent to 
\begin{equation}
\left\{
\begin{array}{l}
q=p_1 =0\cr
c_{11} \neq 0\cr
\end{array}
\right.
\quad \mbox{ or }\quad
\left\{
\begin{array}{l}
q=p_2 = 0\cr
c_{22} \neq 0\cr
\end{array}
\right.
.
\label{sysaltern}
\end{equation}
Both systems $q=p_1=0$ and $q=p_2=0$ generically define 10 points in $\C^2$.
Moreover, for $q=0$, relation $c_{11}\neq 0$ discards the points $P_1,
C$, while relation $c_{22}\neq 0$ discards $P_2,C$. In any case, by
Lemma~\ref{lemma2}, $C$ generically neither belongs to the curve
$p_1=0$, nor to the curve $p_2=0$.

We can prove that
\[
p_1(P_1) = p_2(P_2) = 0.
\]
Let us show only that
\[
p_1(P_1) = 0,
\]
the proof of $p_2(P_2)=0$ being similar.  If $\bxi(P_1)=\bzero$, then the
result holds trivially.  Assume $\bxi(P_1)\neq \bzero$.
We have $q(P_1)=0$, therefore $\angmom_1(P_1) = \angmom_2(P_1)=:\angmom(P_1)$.
Since
generically $P_1\neq C$, by Lemma~\ref{lemma1} we have
$\angmom(P_1)\neq\bzero$, and $\angmom(P_1)$ is parallel to
$\bxi(P_1)$ by Property~\ref{geomrem}.  From $c_{11}(P_1)= 0$ we conclude
that $p_1(P_1) = 0$ because $c_{11}, p_1$ are the projections of
$\angmom, \bxi$ onto $\erho_1$.

Then, we are left with 9 solutions for both systems in
\eqref{sysaltern}, implying a lower bound of 9 solutions for
\eqref{qxi}.

Now we show that \eqref{qp1p2} has exactly 9 solutions.  By Bezout's
theorem we know that it has at most 10 solutions, because both systems $p_1=q=0$ and $p_2=q=0$ have 10 solutions each.  Moreover,
generically we have
\begin{equation}
p_1(P_2)\neq 0, \qquad p_2(P_1)\neq 0.
\label{Pj_not_sol}
\end{equation}
Using
\[
p_1(P_1) = q(P_1) = 0,\qquad p_2(P_2) = q(P_2) = 0
\]
and the lower bound above, we conclude that \eqref{qp1p2} has generically 9
solutions, and the two systems in \eqref{sysaltern} share the same
solutions.

Consider the univariate polynomials
\[
{ \mathfrak{u}_1 = \mathrm{res}(p_1,q,\rho_1), \qquad
  \mathfrak{u}_2 = \mathrm{res}(p_2,q,\rho_1)}
\]
given by the resultant of the pairs $(p_j,q)$ with respect to $\rho_1$, see \cite{cox05}.
%
The quantities $\rho_2'$ and $\rho_2''$ are roots of $\mathfrak{u}_1(\rho_2)$ and
$\mathfrak{u}_2(\rho_2)$ respectively, because they are the $\rho_2$ components of
$P_1$ and $P_2$. By \eqref{Pj_not_sol} we have
\[
\mathfrak{u}_1(\rho_2'')\neq 0,\qquad \mathfrak{u}_2(\rho_2')\neq 0,
\]
therefore $\rho_2'$ and $\rho_2''$ do not solve \eqref{u1u2}.
%
Then we consider
\[
\tilde{\mathfrak{u}}_1 = \frac{\mathfrak{u}_1}{\rho_2-\rho_2'},
  \qquad\tilde{\mathfrak{u}}_2 =
  \frac{\mathfrak{u}_2}{\rho_2-\rho_2''}.
\]
By the previous discussion we must have
\[
{ \tilde{\mathfrak{u}}_1 = c\tilde{\mathfrak{u}}_2,}
\]
with $c$ a non-zero constant, so that the univariate polynomial
\begin{equation}
{ \mathfrak{u} = \mathrm{gcd}(\mathfrak{u}_1,\mathfrak{u}_2)}
\label{upoly}
\end{equation}
has degree 9.
This completes the proof of the theorem.

\rightline{$\square$}

\section{An optimal property of the polynomial $\oldu$}


In Remark~\ref{rem:z} we observed that the Keplerian integrals can be
viewed as polynomials in the variables $\rho,\rhodot,z$ by writing $z$
in place of $\mu/|\erre|$.  Therefore, we can consider the polynomial
system
\begin{equation}
{\angmom_1=\angmom_2,\ \ \mu\tilde\lenz_1=\mu\tilde\lenz_2,\ \
\tilde\energy_1=\tilde\energy_2,\ \ 
z^2_1|\erre_1|^2 = \mu^2,\ \
z^2_2|\erre_2|^2 = \mu^2 ,}
\label{fullsys}
\end{equation} 
of 9 equations in the 6 unknowns
\[
  { \rho_1, \rho_2, \rhodot_1, \rhodot_2, z_1, z_2.}
\]




In the next section we prove that dropping the last two equations in
(\ref{fullsys}) we obtain a consistent polynomial system:
\begin{equation}
\angmom_1=\angmom_2,\quad \mu\tilde{\lenz}_1=\mu\tilde{\lenz}_2,\quad \tilde{\energy}_1=\tilde{\energy}_2.
\label{allbutuj}
\end{equation}
As a consequence of the proof we shall obtain that the univariate polynomial
\[
{ \oldu =
\mathrm{gcd}(\mathfrak{u}_1, \mathfrak{u}_2)}
\] 
of degree 9 has the {\em minimum degree} among the polynomials in
$\rho_1$ or $\rho_2$ contained in the
ideal
\[
{ I =
  \langle\angmom_1-\angmom_2,\ \mu(\tilde{\lenz}_1-\tilde{\lenz}_2),\ \tilde{\energy}_1-\tilde{\energy}_2\rangle\subseteq
  \R[\rho_1,\rho_2,\rhodot_1,\rhodot_2,z_1,z_2]}.
\]

\subsection{A Gr\"obner basis for the ideal $I$}

The following result holds true.

\begin{theorem}  For generic data $\Att_j,\bq_j,\bqdot_j$,
  $j=1,2$, we can find a set of polynomials
\[
\{\mathfrak{g}_1,\ldots,\mathfrak{g}_6\} \subset 
\R[\rho_1,\rho_2,\rhodot_1,\rhodot_2,z_1,z_2]
\]
that is a Gr\"obner basis of the ideal $I$ for the
lexicographic order with
\begin{equation}
\rhodot_1\succ \rhodot_2\succ z_1\succ z_2 \succ \rho_1 \succ \rho_2,
\label{lexord}
\end{equation}
and such that
\[
\mathfrak{g}_6 = \mathfrak{u}.
\]
\label{groebteo}
\end{theorem}

We recall the following definition.
\begin{definition}
  A set $\{\mathfrak{g}_1,\ldots,\mathfrak{g}_n\}$, with $n\in\N$, is
  a {\bf Gr\"obner basis} of a polynomial ideal $I$ for a fixed monomial
  order $\succ$ if and only if the leading term (for that order) of
  any element of $I$ is divisible by the leading term of one
  $\mathfrak{g}_j$.
  \end{definition}

{\bf Proof.}  
For a generic choice of the data we consider the following 
set of generators of $I$:
{
\begin{eqnarray*}
\mathfrak{q}_1 &=& (\angmom_1-\angmom_2)\cdot\DD_1\times\DD_2,\\
\mathfrak{q}_2 &=& (\angmom_1-\angmom_2)\cdot\DD_1\times(\DD_1\times\DD_2),\\
\mathfrak{q}_3 &=& (\angmom_1-\angmom_2)\cdot\DD_2\times(\DD_1\times\DD_2),\\
\mathfrak{q}_4 &= &\mu(\tilde{\lenz}_1-\tilde{\lenz}_2)\cdot\erho_1\times\erho_2,\\
\mathfrak{q}_5 &= &\mu(\tilde{\lenz}_1-\tilde{\lenz}_2)\cdot\DD_1,\\
\mathfrak{q}_6 &= &\mu(\tilde{\lenz}_1-\tilde{\lenz}_2)\cdot\DD_2,\\
\mathfrak{q}_7 &=& \tilde{\energy}_1-\tilde{\energy}_2.
\end{eqnarray*}
}
The first three polynomials have the form
\begin{eqnarray*}
\mathfrak{q}_1 &=& q,\\
\mathfrak{q}_2 &=& |\DD_1\times\DD_2|^2\rhodot_1 - 
\JJ\cdot \DD_1\times(\DD_1\times\DD_2),\\
\mathfrak{q}_3 &=&  |\DD_1\times\DD_2|^2\rhodot_2 - 
\JJ\cdot \DD_2\times(\DD_1\times\DD_2),
\end{eqnarray*}
where $q$ and $\JJ$ are
defined in \eqref{qpoly} and \eqref{JJ}.
The other generators of $I$ can be written as  
\begin{eqnarray*}
\mathfrak{q}_4 &=& -(\DD_1\cdot\erho_2)z_1 -(\DD_2\cdot\erho_1)z_2 + \mathfrak{f}_4,\\
\mathfrak{q}_5 &=& -(\DD_2\cdot\erre_1)z_1 + \mathfrak{f}_5,\\
\mathfrak{q}_6 &=& (\DD_1\cdot\erre_2)z_2 + \mathfrak{f}_6,\\
\mathfrak{q}_7 &=& -z_1 + z_2 + \mathfrak{f}_7,
\end{eqnarray*}
for some polynomials $\mathfrak{f}_j =
\mathfrak{f}_j(\rho_1,\rho_2,\rhodot_1,\rhodot_2)$.
We can substitute $\mathfrak{q}_4,\ldots,\mathfrak{q}_7$ with
\begin{eqnarray*}
  \mathfrak{p}_4 &=& 
  -(\DD_2\cdot\erho_1)\mathfrak{q}_{7}-
  \mathfrak{q}_{4}
  = A z_1 + \mathfrak{a}_1,\\
  \mathfrak{p}_5 &=& (\DD_1\cdot\erho_2)\mathfrak{q}_{7}-
  \mathfrak{q}_{4} = A z_2 + \mathfrak{a}_2,\\
  \mathfrak{p}_6 &=& 
  (\DD_1\cdot\erre_2)\mathfrak{p}_{5}-
  A\mathfrak{q}_6,\\
  \mathfrak{p}_7 &=& 
  (\DD_2\cdot\erre_1)\mathfrak{p}_{4}+
  A\mathfrak{q}_5,
\end{eqnarray*}
where
\[
A= \DD_1\cdot\erho_2 + \DD_2\cdot\erho_1 = (\bq_1-\bq_2)\cdot\erho_1\times\erho_2,
\]
for some polynomials $\mathfrak{a}_j =
\mathfrak{a}_j(\rho_1,\rho_2,\rhodot_1,\rhodot_2)$.
The monomials containing $z_1$, $z_2$ cancel out in
$\mathfrak{p}_{6}$, $\mathfrak{p}_{7}$.

Using $\mathfrak{q}_2=\mathfrak{q}_3=0$, we
eliminate $\rhodot_1$, $\rhodot_2$ from $\mathfrak{p}_4,\ldots,
\mathfrak{p}_7$: we call $\hat{\mathfrak{p}}_4,\ldots,
\hat{\mathfrak{p}}_7$ the polynomials obtained in this way.
It can be shown that
\begin{equation}
{ \hat{\mathfrak{p}}_6 =
-(\DD_1\cdot\erho_2)p_1,\hskip 1cm
\hat{\mathfrak{p}}_7 = 
(\DD_2\cdot\erho_1)p_2,}
\label{p67tilde}
\end{equation}
where $p_1$, $p_2$ are the bivariate polynomials defined in (\ref{p1p2}).

Therefore, the elimination ideal
\[
J := I\cap \R[\rho_1,\rho_2]
\]
is generated by $q,p_1,p_2$:
\[
J = \langle q,p_1,p_2\rangle.
\]
Let us write
\[
q(\rho_1,\rho_2) = \sum_{h=0}^2b_h(\rho_2)\rho_1^h,
\]
with
\[
b_0(\rho_2) = q_{0,2}\rho_2^2 + q_{0,1}\rho_2 + q_{0,0}, \qquad b_1 =
q_{1,0},\qquad b_2=q_{2,0}.
\]
Assuming $q_{2,0}\neq 0$, that generically holds,
let us set
\begin{equation}
  \begin{split}
    &\beta_1=1,\qquad \beta_2 = -\frac{b_1}{b_2}, \qquad
    \gamma_2=-\frac{b_0}{b_2},\\
    &\beta_{h+1}=\beta_h\beta_2+\gamma_h,\qquad
    \gamma_{h+1}=\beta_h\gamma_2,\qquad h=2,3,4.
  \end{split}
  \label{betagamma}
\end{equation}
Moreover we introduce the polynomials
\begin{equation}
  \eta_h(\rho_1) = \frac{1}{b_2}\sum_{j=0}^{h-1}\beta_{h-j}\rho_1^{j},\qquad h=1,\ldots,4.
  \label{etapol}
\end{equation}
With this notation we have
\begin{equation}
\rho_1^{h+1} = \eta_hq + \beta_{h+1}\rho_1 + \gamma_{h+1},\qquad h=1,\ldots,4.
  \label{rho1powers}
\end{equation}
The generators $p_1$, $p_2$ can be written as
\[
p_1(\rho_1,\rho_2) = \sum_{h=0}^4a_{1,h}(\rho_2)\rho_1^h, \qquad
p_2(\rho_1,\rho_2) = \sum_{h=0}^5a_{2,h}(\rho_2)\rho_1^h,
\]
for some polynomials $a_{i,j}$, so that
\[
\tilde{p}_1 = p_1 - q\sum_{j=1}^3a_{1,j+1}\eta_j,\qquad
\tilde{p}_2 = p_2 - q\sum_{j=1}^4 a_{2,j+1}\eta_j
\]
belong to the ideal $J$ and can be written as
\[
\tilde{p}_1 = \tilde{a}_{1,1}(\rho_2)\rho_1 + \tilde{a}_{1,0}(\rho_2), \qquad
\tilde{p}_2 = \tilde{a}_{2,1}(\rho_2)\rho_1 + \tilde{a}_{2,0}(\rho_2),
\]
with
\[
\begin{split}
&\tilde{a}_{1,1} = a_{1,1} + \sum_{h=2}^4a_{1,h}\beta_h, \qquad
  \tilde{a}_{1,0} = a_{1,0} + \sum_{h=2}^4a_{1,h}\gamma_h, \qquad\\
  &\tilde{a}_{2,1} = a_{2,1} + \sum_{h=2}^5a_{2,h}\beta_h, \qquad
  \tilde{a}_{2,0} = a_{2,0} + \sum_{h=2}^5a_{2,h}\gamma_h.
\end{split}
\]
Then we have
\[
J = \langle q,\tilde{p}_1, \tilde{p}_2\rangle.
\]
\smallbreak
Now we set
\[
J_1 =
\langle\tilde{p}_1,\tilde{p}_2\rangle
\]
and prove that 
\[
J = J_1,
\]
that is, we can generate $J$ with two polynomials only.
First we show that
\begin{equation}
V(J_1) = V(J),
\label{equalvariety}
\end{equation}
where the variety $V(K)$ of a polynomial ideal
$K\subseteq\R[\rho_1,\rho_2]$ is the set
\[
V(K) = \{(\rho_1,\rho_2)\in\C^2:p(\rho_1,\rho_2)=0,\ \forall p\in K\}.
\]
From $J_1 \subseteq J$ we have
\begin{equation}
  V(J_1)\supseteq V(J).
\label{VJinVJ1}
\end{equation}
To prove the opposite inclusion, we introduce the univariate
polynomial 
\begin{equation}
\newu = \mathrm{res}(\tilde{p}_1,\tilde{p}_2,\rho_1)
=
\tilde{a}_{1,1}\tilde{a}_{2,0} - \tilde{a}_{1,0}\tilde{a}_{2,1}
\label{vpoly}
\end{equation}
in the variable $\rho_2$.
It turns out that $\newu$ has degree 9.
We need the following results, that hold for a generic choice
of the data:
\begin{itemize}
\item[i)] 
$\oldu$ and $\newu$, defined in \eqref{upoly} and \eqref{vpoly} respectively, have 9 distinct solutions in $\C$ (i.e. they are
  {\em square-free}),
\vskip 0.1cm
\item[ii)] $\tilde{a}_{1,1}$ and $\tilde{a}_{2,1}$ are relatively prime, i.e.
\begin{equation}
  \mathrm{gcd}(\tilde{a}_{1,1},\tilde{a}_{2,1}) = 1.
\label{gcduno}
\end{equation}
\end{itemize}
The proof of these results is in \cite{gbm17}.
By (\ref{gcduno}) we can find two univariate polynomials $\beta,
\gamma$ in the variable $\rho_2$ such that
\begin{equation}
\beta\tilde{a}_{1,1} + \gamma\tilde{a}_{2,1} = 1.
\label{bezout}
\end{equation}
Let us introduce
\begin{equation}
\mathfrak{w} = \beta\tilde{p}_1 + \gamma\tilde{p}_2 = 
  \rho_1 + \mathfrak{z}(\rho_2),
\label{vudoppio}
\end{equation}
where
\[
\mathfrak{z} = \beta\tilde{a}_{1,0} + \gamma\tilde{a}_{2,0}.
\]
We show that
\begin{equation}
J_1 = 
\langle \mathfrak{w},
\newu\rangle.
\label{J1J2}
\end{equation}
In fact
\begin{equation}
\newu = \tilde{a}_{1,1}\tilde{p}_2 - \tilde{a}_{2,1}\tilde{p}_1,
\label{vp1p2}
\end{equation}
because
\[
\tilde{a}_{1,1}\tilde{p}_2 - \tilde{a}_{2,1}\tilde{p}_1 =
\tilde{a}_{1,1}(\tilde{a}_{2,1}\rho_1 + \tilde{a}_{2,0}) -
\tilde{a}_{2,1}(\tilde{a}_{1,1}\rho_1 + \tilde{a}_{1,0}) =
\tilde{a}_{1,1}\tilde{a}_{2,0}  - \tilde{a}_{2,1}\tilde{a}_{1,0}.  
\]
Relations \eqref{vudoppio}, \eqref{vp1p2} show that $\mathfrak{w},
\newu\in J_1$.  On the other hand, inverting these relations we also
obtain
\[
\begin{split}
\tilde{p}_1 &= \tilde{a}_{1,1}\mathfrak{w} - \gamma\newu,\cr
\tilde{p}_2 &= \tilde{a}_{2,1}\mathfrak{w} + \beta\newu,\cr
\end{split}
\]
that is $\tilde{p}_1$, $\tilde{p}_2$ belong to the ideal generated by
$\mathfrak{w}, \newu$.

Property (\ref{J1J2}) implies that $V(J_1)$ has 9 distinct
points.
In fact, for each root $\rho_2$ of $\newu$, which are all distinct because $\newu$ is square-free, we find from
$\mathfrak{w}=0$ a unique value of $\rho_1$ such that $(\rho_1,\rho_2)\in
V(J_1)$.

On the other hand, generically $V(J)$ has 9 distinct points
too.
This can be shown using Theorem~\ref{udeg9} and the fact that 
also $\mathfrak{u}$ is square-free (see \cite{gbm17}).
%
Then, from \eqref{VJinVJ1}
  we have\footnote{Hint: the fact that the variety of two ideals is the same does
  not mean that the two ideals are necessarily the same, see Hilbert's {\em
    nullstellensatz} in \cite{cox05}.}
\begin{equation}
V(J_1) = V(J).
\label{samevar}
\end{equation}
In particular, the polynomials $\newu$ and $\oldu$ coincide up to a
non-zero constant factor $c$:
\[
\newu = c\mathfrak{u},
\]
because their (complex) roots have
the same 9 values.

Now we prove that indeed the two ideals are the same:
\begin{equation}
J_1 = J.
\label{sameideals}
\end{equation}
We only need to show the inclusion $J \subseteq J_1$.
Assume the lexicographic order with
\[
\rho_1 \succ \rho_2
\]
for the monomials in $J$ and take any polynomial $h$ in $J$.  Dividing by $\mathfrak{w} = \rho_1 +
\mathfrak{z}(\rho_2)$ we obtain
\begin{equation}
h(\rho_1,\rho_2) = h_1(\rho_1,\rho_2)\mathfrak{w}(\rho_1,\rho_2)
+ \mathfrak{r}(\rho_2)
\label{decomp}
\end{equation}
for some polynomials $h_1,\mathfrak{r}$. 
The remainder $\mathfrak{r}$ depends only on $\rho_2$ because
of the particular form of $\mathfrak{w}$, whose leading term is $\rho_1$.
From ${\mathfrak{w}\in J_1 \subseteq J}$ and
(\ref{decomp}) we have that ${\mathfrak{r}\in J}$, so that
the roots of $\mathfrak{r}$ must contain all the $\rho_2$ coordinates
of the points in $V(J)$.

Using the fact that $\newu = c\mathfrak{u}$ is square-free we obtain
that $\newu$ must divide $\mathfrak{r}$, i.e. $\mathfrak{r} =
d\newu$ for some polynomal $d(\rho_2)$, which
together with \eqref{decomp} yields
\[
h = h_1\mathfrak{w} + d\newu \in J_1.
\]  
We conclude that \eqref{sameideals} holds.

The polynomials $\mathfrak{g}_1,\ldots,\mathfrak{g}_6$, with
\[{
\mathfrak{g}_1 = \mathfrak{q}_2,\quad
\mathfrak{g}_2 = \mathfrak{q}_3,\quad
\mathfrak{g}_3 = \hat{\mathfrak{p}}_4,\quad
\mathfrak{g}_4 = \hat{\mathfrak{p}}_5,\quad
\mathfrak{g}_5 = \mathfrak{w},\quad
\mathfrak{g}_6 = \mathfrak{u},
}\]
form a Gr\"obner basis of the ideal $I$ for the lexicographic order
(\ref{lexord}).
To show this, we can simply check that the leading monomials
of each pair ($\mathfrak{g}_i, \mathfrak{g}_j$), with $1\leq i<j\leq 6$, are
relatively prime.
This concludes the proof of the theorem.

\rightline{$\square$}

\begin{remark}
The proof above yields a {\em normalized} Gr\"obner basis for the
ideal $J$. In fact, we can rescale by constant factors the polynomials
of the basis and consider
\begin{eqnarray*}
  \mathfrak{g}_1 &=& \rhodot_1 + \mathfrak{h}_1(\rho_1,\rho_2),\\
  \mathfrak{g}_2 &=& \rhodot_2 + \mathfrak{h}_2(\rho_1,\rho_2),\\
  \mathfrak{g}_3 &=& z_1 + \mathfrak{h}_3(\rho_1,\rho_2),\\
  \mathfrak{g}_4 &=& z_2 + \mathfrak{h}_4(\rho_1,\rho_2),\\
  \mathfrak{g}_5 &=& \rho_1 + \mathfrak{z}(\rho_2),\\
    \mathfrak{g}_6 &=& \newu(\rho_2), 
\end{eqnarray*}
with
\[
\mathfrak{h}_1 = \frac{\JJ\cdot\DD_1\times(\DD_1\times\DD_2)}{|\DD_1\times\DD_2|^2},\quad 
\mathfrak{h}_2 = \frac{\JJ\cdot\DD_2\times(\DD_1\times\DD_2)}{|\DD_1\times\DD_2|^2},\quad 
\mathfrak{h}_3 = \frac{\mathfrak{a}_1}{A},\quad
\mathfrak{h}_4 =  \frac{\mathfrak{a}_2}{A}.
\]
\end{remark}
As a consequence of Theorem~\ref{groebteo}, we obtain
\begin{corollary}
  The polynomial $\oldu$ has the minimum degree among the univariate
polynomials in the variable $\rho_2$ belonging to the ideal $I$.
\end{corollary}


\subsection{Selecting the solutions}
\label{s:sel2}

Given $\Avec = (\Att_1,\Att_2)$ with
covariance matrix
\[
{ \Gamma_\Avec =  \left[
\begin{array}{cc}
\Gamma_{\Att_1} &0\cr
0 &\Gamma_{\Att_2}\cr
\end{array}
\right],}
\]
let
\[
  { \Rvec = \Rvec(\Avec)=(\Rcal_1(\Avec),\Rcal_2(\Avec)),
\hskip 0.5cm                                                                      
\Rcal_i = (\rho_i, \rhodot_i),\ \ i=1,2
}
\]
be a solution of
  \begin{equation}
    {
    \bPhi(\Rvec; \Avec) = \left(\begin{array}{c}
\angmom_1-\angmom_2\cr
\bm{\xi}\cdot\erho_1\cr
\end{array}
\right) = \bzero,}
    \label{sistema}
    \end{equation}
where $\bxi$ is defined in \eqref{xidef}, and can also be written as
\begin{equation}
  \bm{\xi} = [\mu(\lenz_1-\lenz_2) - (\energy_1 -
    \energy_2)\erre_1]\times(\erre_1-\erre_2).
  \label{xivec}
\end{equation}
If both  $(\Att_1,\Rcal_1(\Avec))$,
$(\Att_2,\Rcal_2(\Avec))$ give bounded orbits at epochs
\begin{equation}
\tilde t_i = \tilde t_i(\Avec) = \bar t_i -
\frac{\rho_i(\Avec)}{c},\qquad i=1,2,
\label{taberr}
\end{equation}
where aberration of light with velocity $c$ is taken into account, then we can compute the corresponding Keplerian elements.
We introduce the vector
\[{
\bDelta_{a,\ell} =(\Delta a, \Delta\ell),
}\]
representing the difference in semimajor axis and mean anomaly
of the two orbits, comparing the anomalies at the same time $\tilde{t}_1$:
\[
  \Delta a = a_1-a_2, \hskip 0.5cm
    \Delta\ell = \ell_1-\ell_2 - n(a_2)(\tilde{t}_1-\tilde{t}_2),
    \]
    where
$n(a) = \sqrt{\mu} a^{-3/2}$ is the mean motion.
We consider the map
\[
  {
(\Att_1, \Att_2) = \Avec \mapsto \bPsi(\Avec) = \left(\Att_1,
  \Rcal_1,\bDelta_{a,\ell} \right), }
\]
giving the orbit $(\Att_1,\Rcal_1(\Avec))$ in attributable
coordinates at epoch $\tilde t_1$, together with the vector
$\bDelta_{a,\ell}(\Avec)$.

We map the covariance matrix $\Gamma_{\Avec}$ of $\Avec$ into the covariance matrix of
$\bPsi(\Avec)$ by 
\begin{equation*}{
\Gamma_{\bPsi(\Avec)} = \frac{\partial \bPsi}{\partial
\Avec }\; \Gamma_\Avec \; \left[\frac{\partial \bPsi}{\partial
\Avec }\right]^T.
}\end{equation*}

We can consider different ways to select the solutions. Two of them are the following.

\subsubsection{Compatibility conditions.}

We check whether the considered solution of \eqref{sistema} fulfills
the relation
\[
{ \bDelta_{a,\ell} = \bzero}
\]
within a thresold defined by $\Gamma_\Avec$.
More precisely, consider the marginal covariance matrix
\[
\Gamma_{\bDelta_{a,\ell}} = \frac{\partial \bDelta_{a,\ell}}{\partial
\Avec}\Gamma_{\Avec} \left[\frac{\partial \bDelta_{a,\ell}}{\partial
\Avec}\right]^T
\]
of the vector $\bDelta_{a,\ell}$.
The inverse matrix
\[
{ C^{\bDelta_{a,\ell}} = \Gamma^{-1}_{\bDelta_{a,\ell}}}
\]
defines a norm $\Vert\cdot\Vert_{\star}$
allowing to test the identification of $\Att_1, \Att_2$:
\[
  { \Vert \bDelta_{a,\ell}\Vert_{\star}^2 = \bDelta_{a,\ell}
    C^{\bDelta_{a,\ell}} \bDelta_{a,\ell}^T\leq \chi_{max}^2,}
\]
where $\chi_{max}$ is a control parameter, 
that needs to  be seleceted on the basis of            
simulations and practical tests with real data.                                

The orbits computed with the method of Section~\ref{s:kepintlink} are
such that
\begin{equation}
I_1=I_2, \qquad \Omega_1=\Omega_2, \qquad a_1(1-e_1^2) = a_2(1-e_2^2)
\label{eqelem}
\end{equation}
because they fulfill $\angmom_1=\angmom_2$.
Assuming $a_1=a_2$ we get $e_1=e_2$ from the third relation in \eqref{eqelem}.
Since $a_1=a_2$ corresponds to $\energy_1=\energy_2$, from $\bxi=\bzero$ we also obtain
\begin{equation}
\mu(\lenz_1-\lenz_2)\times(\erre_1-\erre_2)=\bzero.
\label{rellenz}
\end{equation}
The vectors $\lenz_1, \lenz_2$ have the same size because $e_1=e_2$. Since it is
quite unlikely that these vector differences are parallel,
generically relation \eqref{rellenz} implies
\[
\omega_1=\omega_2.
\]

\subsubsection{Attribution.}

We can try to attribute the data of ${\cal A}_2$ to each considered
solution $\x_1=({\cal A}_1, {\cal R}_1(\Att))$ of \eqref{sistema}, which has
the covariance matrix
\[{
  \Gamma_{\x_1} =
\left[
\begin{array}{cc}
\Gamma_{\Att_1}  &\Gamma_{\Att_1,\Rcal_1}\cr
\Gamma_{\Rcal_1, \Att_1} &\Gamma_{\Rcal_1}\cr
\end{array}
\right],
}\]
with
\[
\Gamma_{\Att_1} = \frac{\partial \Att_1}{\partial
\Avec}\Gamma_{\Avec} \left[\frac{\partial \Att_1}{\partial
\Avec}\right]^T,
\hskip 0.5cm
\Gamma_{\Rcal_1} = \frac{\partial \Rcal_1}{\partial
\Avec}\Gamma_{\Avec} \left[\frac{\partial \Rcal_1}{\partial
\Avec}\right]^T,
\]
\[
\Gamma_{\Att_1,\Rcal_1} = \Gamma_{\Att_1}\left[\frac{\partial
\Rcal_1}{\partial\Att_1}\right]^T,
\hskip 0.5cm
\Gamma_{\Rcal_1,\Att_1} = \Gamma_{\Att_1,\Rcal_1}^T .
\]
We recall here the attribution algorithm.
Assume that we have
\begin{itemize}
\item[i)] a least squares orbit $\x_1$ obtained from $m_1$ observations,
  with mean epoch $\bar t_1$, with covariance and normal matrices
  $\Gamma_{\x_1}, C_{\x_1}$;
\item[ii)] an attributable ${\cal A}_2$ obtained from
  $m_2$ observations, with mean epoch $\bar t_2$, with covariance and normal matrices
  $\Gamma_{{\cal A}_2}, C_{{\cal A}_2}$.
\end{itemize}

Assume that
\[
  \x
  \mapsto {\cal A} = G(\x)
\]
maps orbital elements to attributables and 
define the prediction function
\[
F(\x;t_0,t)=G\circ\Phi_{t_0}^t(\x),
\]
where $\Phi_{t_0}^t(\x)$ is the integral flow of the Kepler problem.
The covariance and normal matrices of ${\cal A}$
are given by
\[
\Gamma_{\cal A}=\left[\frac{\partial F}{\partial \x}\right] \Gamma_{\x} \left[\frac{\partial F}{\partial \x}\right]^T, \qquad C_{\cal A}= \Gamma_{\cal A}^{-1},
\]
where $\Gamma_{\x}$ is the covariance matrix of $\x$.


Let ${\cal A}_2$ be an attributable and $C_2$ its
$4\times 4$ normal matrix.
Let ${\cal A}_p$ be the predicted attributable at time $\bar t_2$, computed from the
least squares orbit $\x_1$, and $\Gamma_p, C_p$ its covariance and
normal matrices.
\medbreak
The formulae for linear attribution in the 4-D space are the following (see \cite{mg10}):
\begin{equation}
    \begin{split}
C_0&=C_2+C_p,\qquad \Gamma_0=C_0^{-1},\nonumber\\
\x_0&= \Gamma_0\left[C_2 {\cal A}_2 + C_p {\cal A}_p\right],
\nonumber\\
K_4&= ({\cal A}_p - {\cal A}_2)\cdot \left [C_2-C_2\,\Gamma_0\,C_2\right]\,({\cal A}_p - {\cal A}_2).
\nonumber
    \end{split}
\end{equation}
The values of the attribution penalty $K_4/m$, with $m=m_1+m_2$, is
used to filter out the pairs orbit-attributable which cannot belong to
the same object.

\subsection{Numerical test with {\tt Link2}}

We show an application of the {\tt Link2} algorithm using 4
observations of asteroid (4542) {\em Mossotti} made on April 28, 2011 and 4
observations of the same asteroid made on November 4, 2013. These data
have been collected by the telescope Pan-STARRS1, mount Hakeakala,
Hawaii, and are displayed in Table~\ref{tab:obs}. For simplicity, only
a few digits are reported here.  \small
\begin{table}[h!]
    \centering
      \begin{tabular}{p{0.15\textwidth}|p{0.15\textwidth}|p{0.15\textwidth}}
      \hline
      $\alpha$ (rad) &$\delta$ (rad) &$t$ (MJD)\\
      \hline
      4.127300       &$-$0.094246    &55679.51169  \\  
      4.127261       &$-$0.094238    &55679.52398  \\    
      4.127221       &$-$0.094230    &55679.53664   \\   
      4.127188       &$-$0.094223    &55679.54709  \\   
      \hline
      0.896220     &\ 0.078635    &56600.43378 \\
      0.896168     &\ 0.078626    &56600.44773 \\
      0.896119     &\ 0.078617    &56600.46130 \\
      0.896069     &\ 0.078608    &56600.47489 \\    
    \end{tabular}
      \vskip 0.1cm
    \caption{Values of right ascension ($\alpha$) and declination ($\delta$) used for the linkage.}
  \label{tab:obs}
\end{table}
\normalsize
From these observations we computed the attributables
  \[\begin{split}
  &{\cal A}_1 = (4.127242, -0.094234, -0.00316982,  \phantom{-}0.00064761),\\
  &{\cal A}_2 = (0.896144, \phantom{-}0.078622, -0.00364403, -0.00065882),
  \end{split}
  \]
at the mean epochs $\bar{t}_1=55679.52985$ MJD, $\bar{t}_2=56600.45442$ MJD. In the attributables ${\cal A}_1$, ${\cal A}_2$ the angles $\alpha$, $\delta$ are given in radians and the angular rates $\alphadot$, $\deltadot$ are given in radians/day.
  

After discarding solutions with non-real or non-positive values of
$\rho$, and unbounded solutions, we are left with the radial distance
pair
\[
  (\rho_1,\rho_2) = (1.8802, 2.1774) \ \mathrm{au},
\]
leading to the pair of preliminary orbits given in Table~\ref{tab:prelim}.
\begin{table}[h!]
  \begin{center}
\begin{tabular}{p{0.12\textwidth}|p{0.12\textwidth}|p{0.13\textwidth}|p{0.13\textwidth}|p{0.13\textwidth}|p{0.13\textwidth}|p{0.15\textwidth}}
      \hline
  $a$ (au)  &$e$  &$I$  &$\Omega$  &$\omega$  &$\ell$  &$\tilde{t}$ (MJD) \cr
  \hline
  3.03055 &0.06436   &11.22246  &104.80204  &117.44122   &\quad 5.63111  &55679.51899\cr
  3.02287 &0.04015   &11.22246  &104.80204  &114.03999  &188.86754  &56600.44185\cr
\end{tabular}
  \end{center}
  \caption{Pair of preliminary orbits computed with {\tt Link2}. The epoch $\tilde{t}$ has been corrected by aberration, see \eqref{taberr} in Section~\ref{s:sel2}. The angles $I, \Omega, \omega, \ell$ are given in degrees.}
  \label{tab:prelim}
\end{table}

The intersection of the curves defined by $p_1=p_2=q=0$ is shown in Figure~\ref{fig:p1p2q}.
\begin{figure}
\centerline{\epsfig{figure=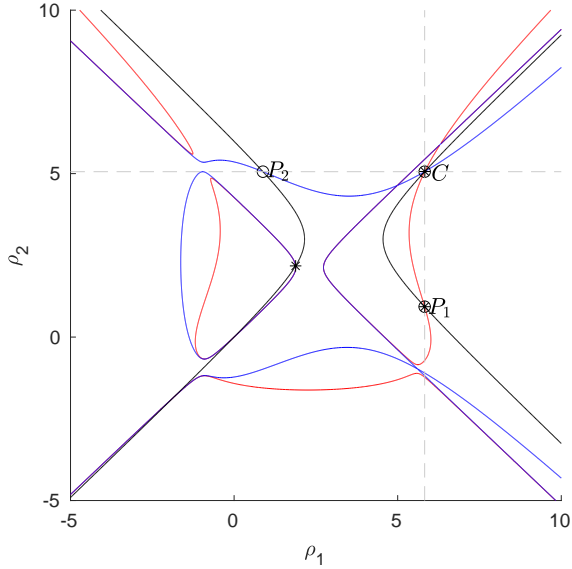,width=7.5cm}}
\caption{Intersections of the curves $p_1=0$ (red), $p_2=0$ (blue), $q=0$ (black) in the $\rho_1 \rho_2$ plane.}
\label{fig:p1p2q}
\end{figure}

Then we computed the rms of the preliminary orbits in
Table~\ref{tab:prelim} with respect to a pure Keplerian motion and
selected the first orbit as the best (the one with the least rms).  We
propagated this orbit at the mean epoch of the observations, which is
$\bar{t}=56139.99213$,
applied differential corrections and computed a least squares orbit.
This orbit is shown in Table~\ref{tab:LS}, together with the known orbit at the same epoch.\footnote{data from AstDyS-2 ({\tt https://newton.spacedys.com/astdys/}), orbit propagation with the OrbFit software ({\tt http://adams.dm.unipi.it/orbfit/})}
\begin{table}[h!]
  \begin{center}
    \begin{tabular}{p{0.08\textwidth}|p{0.12\textwidth}|p{0.12\textwidth}|p{0.13\textwidth}|p{0.13\textwidth}|p{0.13\textwidth}|p{0.13\textwidth}}
      \hline
  &$a$ (au)  &$e$  &$I$  &$\Omega$  &$\omega$  &$\ell$  \cr 
  \hline          
  LS   &3.01802    &0.05755  &11.32849  &104.37041   &146.76038  &66.54688\cr  
  known  &3.00997  &0.05614  &11.30734  &104.41991   &144.01204  &69.25283
\end{tabular}
  \end{center}
  \caption{Orbital elements of the least squares solution (LS) and of the known orbit. The angles $I, \Omega, \omega, \ell$ are given in degrees.}
  \label{tab:LS}
\end{table}

\section{Joining three TSAs}

Given three TSAs with attributables ${\cal A}_1, {\cal A}_2, {\cal
  A}_3$ at mean epochs $\bar{t}_1, \bar{t}_2, \bar{t}_3$, setting the
conservation of angular momentum is enough to obtain a finite number
of orbits.  We review the following result, presented in
\cite{gbm17}. Here the subscripts in $\angmom_i, \rho_i, \rhodot_i,
\DD_i, \EE_i, \FF_i, \GG_i$ refer to the three epochs.

\begin{proposition}
Assume
\begin{equation}
  \DD_1\times\DD_2\cdot\DD_3 \neq 0.
\label{D1D2D3}
\end{equation}
Then the polynomial system
\begin{subequations}
\begin{eqnarray}
(\angmom_1-\angmom_2)\cdot\DD_1\times\DD_2 &=& 0,\label{AM12uno}\\
(\angmom_1-\angmom_2)\cdot\DD_1\times(\DD_1\times\DD_2) &=& 0,\label{AM12due}\\
(\angmom_2-\angmom_3)\cdot\DD_2\times\DD_3 &=& 0,\label{AM23uno}\\
(\angmom_2-\angmom_3)\cdot\DD_2\times(\DD_2\times\DD_3) &=& 0,\label{AM23due}\\
(\angmom_3-\angmom_1)\cdot\DD_3\times\DD_1 &=& 0,\label{AM31uno}\\
(\angmom_3-\angmom_1)\cdot\DD_3\times(\DD_3\times\DD_1) &=& 0\label{AM31due}
\end{eqnarray}
\label{AMprojections}
\end{subequations}
in the 6 unknowns
\[
  \rho_1,\rhodot_1,\rho_2,\rhodot_2,\rho_3,\rhodot_3
  \]
  is equivalent to the redundant system
\begin{equation}
  \angmom_1 = \angmom_2,
  \qquad
  \angmom_2 = \angmom_3,
  \qquad
  \angmom_3 = \angmom_1.
  \label{angmomeqs}
\end{equation}
\end{proposition}

\begin{proof}
  System \eqref{angmomeqs} trivially implies
  \eqref{AMprojections}. Assume now that system \eqref{AMprojections}
  holds. Using relations \eqref{AM31uno}, \eqref{AM31due}, to prove
  that $\angmom_3=\angmom_1$ we only need to show that
\begin{equation}
(\angmom_3-\angmom_1)\cdot\bv = 0
  \label{c3mc1_v}
  \end{equation}
for some vector $\bv$ that does not belong to the linear space generated by
$\DD_3\times\DD_1$ and $\DD_3\times(\DD_3\times\DD_1)$.
Indeed we show that we can choose
\[
\bv = \DD_1\times\DD_2.
\]
Note that
\[
(\DD_1\times\DD_2) \cdot 
(\DD_2\times\DD_3)\times\bigl(\DD_2\times(\DD_2\times\DD_3)\bigr) =  0,
\]
that is, the vector $\DD_1\times\DD_2$ belongs to the
linear space generated by $\DD_2\times\DD_3$ and
$\DD_2\times(\DD_2\times\DD_3)$.
Moreover, $\DD_1\times\DD_2$ is not generated by $\DD_3\times\DD_1$ and
$\DD_3\times(\DD_3\times\DD_1)$, in fact by (\ref{D1D2D3}) we have
\[
(\DD_1\times\DD_2) \cdot 
(\DD_3\times\DD_1)\times\bigl(\DD_3\times(\DD_3\times\DD_1)\bigr) = 
|\DD_3\times\DD_1|^2 \DD_1\times\DD_2\cdot\DD_3 \neq 0.
\]
Setting
\[
\bv = \DD_1\times\DD_2,
\]
from \eqref{AM12uno}, \eqref{AM23uno}, \eqref{AM23due} we obtain
$(\angmom_1-\angmom_2)\cdot\bv = (\angmom_2-\angmom_3)\cdot\bv = 0$,
which yield \eqref{c3mc1_v} and therefore we obtain
$\angmom_3=\angmom_1$.  In a similar way we can prove that
$\angmom_1=\angmom_2$, $\angmom_2=\angmom_3$, provided that system
\eqref{AMprojections} holds.

\rightline{$\square$}
\end{proof}

Equations \eqref{angmomeqs}
can be written as
\begin{eqnarray*}
&&\DD_1\rhodot_1 - \DD_2\rhodot_2 = \JJ_{12}(\rho_1,\rho_2),\hskip 0.5cm\\
&&\DD_2\rhodot_2 - \DD_3\rhodot_3 = \JJ_{23}(\rho_2,\rho_3),\hskip 0.5cm\\
&&\DD_3\rhodot_3 - \DD_1\rhodot_1 = \JJ_{31}(\rho_3,\rho_1),
%
\end{eqnarray*}
where
\begin{eqnarray*}
\JJ_{12}(\rho_1,\rho_2) &=& \EE_2\rho_2^2 - \EE_1\rho_1^2 + \FF_2\rho_2 -
\FF_1\rho_1 + \GG_2 - \GG_1,\\
\JJ_{23}(\rho_2,\rho_3) &=& \EE_3\rho_3^2 - \EE_2\rho_2^2 + \FF_3\rho_3 -
\FF_2\rho_2 + \GG_3 - \GG_2,\\
\JJ_{31}(\rho_3,\rho_1) &=& \EE_1\rho_1^2 - \EE_3\rho_3^2 + \FF_1\rho_1 -
\FF_3\rho_3 + \GG_1 - \GG_3.
\end{eqnarray*}

Equations (\ref{AM12uno}), (\ref{AM23uno}), (\ref{AM31uno})
depend only on the radial distances. In fact, they correspond to the system
\begin{equation}
{\JJ_{12}\cdot \DD_1\times\DD_2 = 0,
\quad
\JJ_{23}\cdot \DD_2\times\DD_3 = 0,
\quad
\JJ_{31}\cdot \DD_3\times\DD_1 = 0,}
\label{JJ123}
\end{equation}
which can be written as
\begin{eqnarray}
q_3 &=& a_3\rho_2^2 + b_3\rho_1^2 + c_3\rho_2 + d_3\rho_1 + e_3 = 0,
\label{qtre}\\
q_1 &=& a_1\rho_3^2 + b_1\rho_2^2 + c_1\rho_3 + d_1\rho_2 + e_1 = 0,
\label{quno}\\
q_2 &=& a_2\rho_1^2 + b_2\rho_3^2 + c_2\rho_1 + d_2\rho_3 + e_2 = 0,\label{qdue}
\end{eqnarray}
where
\begin{eqnarray*}
&&a_3 =  \EE_2\cdot \DD_1\times\DD_2,\qquad
b_3 = -\EE_1\cdot \DD_1\times\DD_2,\\
&&c_3 =  \FF_2\cdot \DD_1\times\DD_2,\qquad
d_3 = -\FF_1\cdot \DD_1\times\DD_2,\\
&&\hskip 1cm e_3 =  (\GG_2-\GG_1)\cdot \DD_1\times\DD_2,
\end{eqnarray*}
and the other coefficients $a_j, b_j, c_j, d_j, e_j$, for $j=1,2$, have
similar expressions, obtained by cycling the indexes.

To eliminate $\rho_1, \rho_3$ from (\ref{JJ123}) we can first compute the
resultant
\[
r = \mathrm{res}(q_3,q_2,\rho_1),
\]
which depends only on $\rho_2,\rho_3$, and then the resultant
\[
\mathfrak{q} = \mathrm{res}(r,q_1,\rho_3),
\]
which is a univariate polynomial of degree 8 in the variable $\rho_2$.

Therefore, provided that (\ref{D1D2D3}) holds, to get the
solutions of (\ref{angmomeqs}) we search for the roots $\bar{\rho}_2$
of $\mathfrak{q}(\rho_2)$, compute the corresponding values
$\bar{\rho}_3$ of $\rho_3$ from $r(\rho_3,\bar{\rho}_2) =
q_1(\rho_3,\bar{\rho}_2) = 0$, and the values
of $\rho_1$ from $q_3(\rho_1,\bar{\rho}_2) =
q_2(\bar{\rho}_3,\rho_1) = 0.$
  
From equations \eqref{AM12due}, \eqref{AM23due}, \eqref{AM31due}
we can write the radial velocities $\rhodot_j$ as functions of pairs of
radial distances:
{  \begin{eqnarray*}
\rhodot_2 &=&
\frac{\JJ_{12}(\rho_1,\rho_2)\cdot\DD_1\times(\DD_1\times\DD_2)}{|\DD_1\times\DD_2|^2},\\
\rhodot_3 &=&
\frac{\JJ_{23}(\rho_2,\rho_3)\cdot\DD_2\times(\DD_2\times\DD_3)}{|\DD_2\times\DD_3|^2},\\
\rhodot_1 &=&
\frac{\JJ_{31}(\rho_3,\rho_1)\cdot\DD_3\times(\DD_3\times\DD_1)}{|\DD_3\times\DD_1|^2}.
\end{eqnarray*}}
From these data we can reconstruct the orbital elements.


\subsection{Straight line solutions}
  
A particular solution of system (\ref{angmomeqs}) can be obtained by
searching for values of $\rho_j,\rhodot_j$ such that
\[
{ \angmom_j(\rho_j,\rhodot_j) = \bzero, \hskip 1cm j=1,2,3.}
\]
Let us drop the index $j$.
Relation $\erre\times\erredot = \bzero$ implies that there exists
$\lambda\in\R$ such that
\begin{equation}
{ \rhodot\erho + \rho\etabf + \bqdot = \lambda(\rho\erho + \bq),}
\label{parall}
\end{equation}
with $\etabf = \alphadot\cos\delta\ealpha + \deltadot\edelta$.
Setting $\sigma = \rhodot-\lambda\rho$ we can write (\ref{parall}) as
\begin{equation}
{ \sigma\erho + \rho\etabf - \lambda\bq = -\bqdot.}
\label{parall2}
\end{equation}
We introduce the vector
\[
{ \bu = \bq - (\bq\cdot\erho)\erho - \frac{1}{\eta^2}(\bq\cdot\etabf)\etabf,}
\]
which is orthogonal to both $\erho$ and $\etabf$, where $\eta = |\etabf|$.
%

Thus, we can write (\ref{parall2}) as
\[{
[\sigma - \lambda(\bq\cdot\erho)]\erho +
\Bigl[\rho-\frac{\lambda}{\eta^2}(\bq\cdot\etabf)\Bigr]\etabf - \lambda\bu =
-\bqdot.
}\]
Since $\{\erho,\etabf,\bu\}$ is generically an orthogonal basis of $\R^3$, we
find
\[{
\lambda = \frac{1}{|\bu|^2}(\bqdot\cdot\bu),\qquad
\rho = \frac{1}{\eta^2}(\lambda\bq - \bqdot)\cdot\etabf,\qquad
\rhodot = \lambda\rho + (\lambda\bq - \bqdot)\cdot\erho.
}\]

\noindent In particular, we obtain the value 
\[
{   \rho = \frac{1}{\eta^2}
  \Bigl(\frac{1}{|\bu|^2}(\bqdot\cdot\bu)(\bq\cdot\etabf) - 
    \bqdot\cdot\etabf\Bigr)}
\]
for the radial distance, corresponding to a solution with zero angular
momentum.

\subsection{Selecting the solutions}

Given $\Avec = (\Att_1,\Att_2,\Att_3)$ with
covariance matrices $\Gamma_{\Att_1}, \Gamma_{\Att_2}, \Gamma_{\Att_3}$, let
\[
\Rvec = \Rvec(\Avec) = \bigl({\cal R}_1(\Avec), {\cal R}_2(\Avec), {\cal R}_3(\Avec)\bigr), \qquad {\cal R}_i=(\rho_i,\rhodot_i), \quad i=1,2,3
\]
be a solution of
\begin{equation}
\bPhi(\Rvec;\Avec) = {\bf 0},
\label{Phi_eq_0}
\end{equation}
with
\[
\bPhi(\Rvec; \Avec) = \left(
\begin{array}{c}
(\angmom_1 - \angmom_2)\cdot\DD_1\times(\DD_1\times\DD_2)\cr
(\angmom_1 - \angmom_2)\cdot\DD_1\times\DD_2\cr
(\angmom_2 - \angmom_3)\cdot\DD_2\times(\DD_2\times\DD_3)\cr
(\angmom_2 - \angmom_3)\cdot\DD_2\times\DD_3\cr
(\angmom_3 - \angmom_1)\cdot\DD_3\times(\DD_3\times\DD_1)\cr
(\angmom_3 - \angmom_1)\cdot\DD_3\times\DD_1\cr
\end{array}
\right) .
\]
%
If $({\cal A}_1, {\cal R}_1(\Att))$, $({\cal A}_2, {\cal R}_2(\Att))$, and
$({\cal A}_3, {\cal R}_3(\Att))$ give bounded orbits at epochs
\[
\tilde{t}_i = \bar{t}_i-\frac{\rho_i(\Att)}{c}, \qquad i=1,2,3,
\]
then we compute the corresponding Keplerian elements.
We introduce the difference vectors
\begin{eqnarray*}
\bDelta_{12} &=& \bigl(a_1-a_2, \omega_1-\omega_2 , 
\ell_1 - \ell_2 - n(a_2)(\tilde{t}_1-\tilde{t}_2)\bigr),\\
\bDelta_{32} &=&\bigl(a_3-a_2, \omega_3-\omega_2, 
\ell_3- \ell_2 - n(a_2)(\tilde{t}_3-\tilde{t}_2)\bigr),
\end{eqnarray*}
where
$n(a) = \sqrt{\mu} a^{-3/2}$ is the mean motion.
%
%
We consider map
\[
(\Att_1, \Att_2 ,\Att_3) = \Avec \mapsto \bPsi(\Avec) = \left(\Att_2,
  \Rcal_2,\bDelta_{12}, \bDelta_{32}\right),
\]
giving the orbit $(\Att_2,\Rcal_2)$ in attributable coordinates at epoch
$\tilde{t}_2$ together with the vectors $\bDelta_{12}$,
$\bDelta_{32}$, which are not constrained by the angular momentum
integrals.

We map the covariance matrix
\[
\Gamma_{\Avec} = \left[\begin{array}{ccc}
    \Gamma_{{\cal A}_1} &0 &0\cr
    0 &\Gamma_{{\cal A}_1} &0 \cr
    0 &0 &\Gamma_{{\cal A}_1}\cr
  \end{array}
  \right]
\]
of $\Avec$ into the covariance matrix of
$\bPsi(\Avec)$ by the covariance propagation rule:
\begin{equation*}
\Gamma_{\bPsi(\Avec)} = \frac{\partial \bPsi}{\partial
\Avec }\; \Gamma_\Avec \; \left[\frac{\partial \bPsi}{\partial
\Avec }\right]^T,
\end{equation*}

We can check whether the considered solution of (\ref{Phi_eq_0})
fulfills the {\bf compatibility conditions}
\[
\bDelta_{12} = \bDelta_{32} = \bzero
\]
within a threshold defined by $\Gamma_\Att$.
%
More precisely, consider the marginal covariance matrix
$\Gamma_{\bDelta}$
of the vector
\[
\bDelta = (\bDelta_{12}, \bDelta_{32}).
\]
The inverse matrix $C^{\bDelta} = \Gamma^{-1}_{\bDelta}$ defines a norm
$\Vert\cdot\Vert_{\star}$ 
allowing us to test an
identification between the attributables $\Att_1, \Att_2, \Att_3$: we check
whether 
\begin{equation}
\Vert \bDelta\Vert_{\star}^2 = \bDelta C^{\bDelta}
\bDelta^T\leq \chi_{max}^2,
\label{control_3}
\end{equation}
where $\chi_{max}$ is a control parameter.


\subsection{Numerical test with {\tt Link3}}

We show an application of the {\tt Link3} algorithm using three TSAs of observations of asteroid (4628) {\em Laplace}, listed in Table~\ref{tab:obs3}.
\small
\begin{table}[h!]
    \centering
      \begin{tabular}{p{0.15\textwidth}|p{0.15\textwidth}|p{0.15\textwidth}}
      \hline
      $\alpha$ (rad) &$\delta$ (rad) &$t$ (MJD)\\
      \hline
      5.497381      &$-$0.067942    &55794.33902\\
      5.497339      &$-$0.067950    &55794.35011\\
      5.497195      &$-$0.067978    &55794.38807\\
      5.497148      &$-$0.067987    &55794.40021\\
      \hline
      0.715965      & 0.542095    &56226.52009\\
      0.715918      & 0.542080    &56226.53117\\
      0.715867      & 0.542063    &56226.54334\\
      0.715816      & 0.542047    &56226.55525\\
      \hline
      0.831317      & 0.390743    &56358.23971\\
      0.831350      & 0.390746    &56358.24497\\
      0.831383      & 0.390749    &56358.25023\\
      0.831416      & 0.390751    &56358.25550\\
      \end{tabular}
      \vskip 0.1cm
    \caption{Values of right ascension ($\alpha$) and declination ($\delta$) of asteroid (4628) {\em Laplace} collected by the Pan-STARRS1 telescope.}
  \label{tab:obs3}
\end{table}
\normalsize
From these observations we computed the three attributables
\[
\begin{split}
  &{\cal A}_1 = (5.497266, -0.067965, -0.00379969, -0.00072536)\\
  &{\cal A}_2 = (0.715891, \phantom{-}0.542071, -0.00422693, -0.00136864)\\
  &{\cal A}_3 = (0.831367, \phantom{-}0.390747, \phantom{-}0.00622482,  \phantom{-}0.00054073)
\end{split}
\]
at the mean epochs $\bar{t}_1=55794.36935$, $\bar{t}_2=56226.53746$, $\bar{t}_3=56358.24760$, given in MJD. In the attributables ${\cal A}_j$ the angles are given in radians and the angular rates in radians/day.
After discarding the straight-line solution, the solutions with
negative values of $\rho$, and the unbounded ones, we are left with the radial distance
triplets
\[
\begin{split}
&(\rho_1,\rho_2,\rho_3) = (2.1955, 1.9028, 2.9200)\ \mathrm{au},\\
&(\rho_1,\rho_2,\rho_3) = (1.9379, 1.8279, 2.8870)\ \mathrm{au},
\end{split}
\]
leading to the triplets of preliminary orbits displayed in
Table~\ref{tab:prelim3}. \small
\begin{table}[h!]
  \begin{center}
\begin{tabular}{p{0.02\textwidth}|p{0.12\textwidth}|p{0.12\textwidth}|p{0.12\textwidth}|p{0.14\textwidth}|p{0.14\textwidth}|p{0.14\textwidth}|p{0.15\textwidth}}
      \hline
  &$a$ (au)  &$e$  &$I$  &$\Omega$  &$\omega$  &$\ell$  &$\tilde{t}$ (MJD) \cr
  \hline
 &2.86808    &0.30942   &12.13274  &274.68641  &172.31982  &266.26844  &55794.35667\cr
1 &2.64520   &0.13981   &12.13274  &274.68641  &258.53770  &242.07553  &56226.52647\cr
  &2.59619   &0.03219   &12.13274  &274.68641  &290.50786  &228.16130  &56358.23074\cr
\hline
 &2.64614    &0.11646   &11.78916  &275.69255  &249.45265  &149.80066  &55794.35816\cr
2 &2.64562    &0.11562  &11.78916  &275.69255  &248.51598  &249.78277  &56226.52691\cr
  &2.64427    &0.11343   &11.78916  &275.69255  &247.58320  &280.66987   &56358.23093\cr
\end{tabular}
  \end{center}
  \caption{Triplets of preliminary orbits computed with {\tt Link3}. The angles $I, \Omega, \omega, \ell$ are given in degrees.}
  \label{tab:prelim3}
\end{table}
\normalsize

Based on the norm $\Vert\bDelta\Vert_\star$, we selected the second triplet.
Checking the rms of these orbits with respect to a pure Keplerian
motion we selected the first orbit of this triplet.
We propagated this orbit at the mean epoch of the 12 observations in
Table~\ref{tab:obs3}, which is $\bar{t} =56126.38480$, applied
differential corrections and computed a least squares orbit. This
orbit is shown in Table~\ref{tab:LS3}, toghether with the known orbit
at the same epoch.

\begin{table}[h!]
  \begin{center}
    \begin{tabular}{p{0.08\textwidth}|p{0.12\textwidth}|p{0.12\textwidth}|p{0.14\textwidth}|p{0.14\textwidth}|p{0.14\textwidth}|p{0.14\textwidth}}
      \hline
  &$a$ (au)  &$e$  &$I$  &$\Omega$  &$\omega$  &$\ell$  \cr 
  \hline          
  LS     &2.64443  &0.11729  &11.79295  &275.66956  &248.45817  &227.09536\cr
  known  &2.64441  &0.11730  &11.79294  &275.66961  &248.46069  &227.09295\cr
\end{tabular}
  \end{center}
  \caption{Orbital elements of the least squares solution (LS) and of the known orbit. The angles $I, \Omega, \omega, \ell$ are given in degrees.}
  \label{tab:LS3}
\end{table}

\section{Conclusions and future work}

We reviewed two initial orbit determination methods for TSAs of
optical observations employing the conservation laws of Kepler's
problem. Some algebraic properties of these algorithms have also been
discussed and a simple test case has been presented for both.
Being based on conservation laws, these methods are suitable to link
TSAs quite far apart in time, even differing by more than one orbital
period of the observed body.  Moreover, these algorithms are very
fast, because they are based on a polynomial formulation with low
degree (9 for {\tt Link2}, 8 for {\tt Link3}). The sensitivity of
these algorithms to astrometric errors is an important feature to be
investigated: in fact it seems that some orbital elements are more
sensitive to these errors.  Moreover, it would be important to find
efficient filters to discard {\em a priori} pairs of TSAs that are not
likely to belong to the same observed object.  Indeed, even if some
filters have been proposed in \cite{gdm10}, \cite{gbm15}, a
satisfactory solution to this problem is still missing. The mentioned
problems are currently under investigation.

\section{Acknowledgements}

The author wishes to thank dr. Giulio Ba\`u for carefully reading
the manuscript, and for his useful suggestions. The author
has been partially supported by the MSCA-ITN Stardust-R, Grant
Agreement n. 813644 under the H2020 research and innovation program.
He also acknowledges the project MIUR-PRIN 20178CJA2B ``New frontiers
of Celestial Mechanics: theory and applications'' and the GNFM-INdAM
(Gruppo Nazionale per la Fisica Matematica).

\bibliography{mybib}{}

\begin{thebibliography}{10}

\bibitem{cellpinz}
A.~{Celletti} and G.~{Pinzari}.
\newblock {Four classical methods for determining planetary elliptic elements:
  a comparison}.
\newblock {\em Cel. Mech. Dyn. Ast.}, 93:1--52, 2005.

\bibitem{cox05}
D.~{Cox}, J.~{Little}, and D.~{O'Shea}.
\newblock {\em {Ideals, Varieties, and Algorithms}}.
\newblock Springer, 2005.

\bibitem{gauss1809}
C.~F. {Gauss}.
\newblock {\em {Theoria motus corporum in sectionibus conicis solem
  ambientium}}.
\newblock reprinted by Dover publications in 1963, 1809.

\bibitem{gbm15}
G.~F. {Gronchi}, G.~{Ba\`u}, and S.~{Mar\`o}.
\newblock {Orbit determination with the two-body integrals. III}.
\newblock {\em Cel. Mech. Dyn. Ast.}, 123/2:105--122, 2015.

\bibitem{gbm17}
G.~F. {Gronchi}, G.~{Ba\`u}, and A.~{Milani}.
\newblock {Keplerian integrals, elimination theory and identification of very
  short arcs in a large database of optical observations}.
\newblock {\em Cel. Mech. Dyn. Ast.}, 127/2:211--232, 2017.

\bibitem{gbrjm21}
G.~F. {Gronchi}, G.~{Ba\`u}, \'O. {Rodr\'iguez}, R.~{Jedicke}, and
  J.~{Moeyens}.
\newblock {A generalization of a method by Mossotti for initial orbit
  determination}.
\newblock {\em Cel. Mech. Dyn. Ast.}, 133:41, 2021.

\bibitem{gdm10}
G.~F. {Gronchi}, L.~{Dimare}, and A.~{Milani}.
\newblock {Orbit determination with the two-body integrals}.
\newblock {\em Cel. Mech. Dyn. Ast.}, 107/3:299--318, 2010.

\bibitem{gfd11}
G.~F. {Gronchi}, D.~{Farnocchia}, and L.~{Dimare}.
\newblock {Orbit determination with the two-body integrals. II}.
\newblock {\em Cel. Mech. Dyn. Ast.}, 110/3:257--270, 2011.

\bibitem{laplace}
P.~S. {Laplace}.
\newblock {\em M\'em. Acad. R. Sci. Paris}, 10:93--146, 1780.

\bibitem{mg10}
A.~{Milani} and G.~F. {Gronchi}.
\newblock {\em {Theory of Orbit Determination}}.
\newblock Cambridge Univ. Press, 2010.

\bibitem{taff84}
L.~G. {Taff}.
\newblock {On initial orbit determination}.
\newblock {\em The Astronomical Journal}, 89(6):1426--1428, 1984.

\bibitem{th77}
L.~G. {Taff} and D.~L. {Hall}.
\newblock {The use of angles and angular rates. I - Initial orbit
  determination}.
\newblock {\em Celestial Mechanics}, 16:481--488, 1977.

\bibitem{trs84}
L.~G. {Taff}, P.~M.~S. {Randall}, and S.~A. {Stansfield}.
\newblock {Angles-Only, Ground-Based, Initial Orbit Determination}.
\newblock {\em technical report, Lincoln Laboratory}, 1984.

\end{thebibliography}
\bibliographystyle{plain}

\end{document}